\documentclass[aps,twocolumn,nofootinbib,preprintnumbers,superscriptaddress,citeautoscript]{revtex4-1}
\pdfoutput=1

\usepackage{graphicx}
\usepackage{epsfig,amsmath,amsthm,amsfonts,mathrsfs,amssymb}
\usepackage[usenames,dvipsnames,table]{xcolor}
\usepackage{hyperref}

\usepackage{amsfonts}
\usepackage{braket}
\usepackage{bbold}
\usepackage{mathrsfs}  
\usepackage[mathscr]{eucal}
\usepackage{enumerate}
\usepackage[normalem]{ulem}
\usepackage{bm}

\hypersetup{
    pdfstartview={FitH},    % fits the width of the page to the window
    pdftitle={},    % title
    pdfauthor={Sergio Hern\'andez Cuenca, Veronika  E. Hubeny, Mukund Rangamani, Massimiliano Rota},     % author
    colorlinks=true,       % false: boxed links; true: colored links
    linkcolor=blue,          % color of internal links (change box color with linkbordercolor)
    citecolor=red,        % color of links to bibliography
    filecolor=magenta,      % color of file links
    urlcolor=blue           % color of external links
}

\newtheorem{lemma}[]{Lemma}
\newtheorem{definition}[]{Definition}

\newtheorem{problem}[]{Problem}

\def\a{\mathcal{A}}
\def\b{\mathcal{B}}
\def\c{\mathcal{C}}

\def\I{{\bf I}}

\def\N{{\sf N}}
\def\D{{\sf D}}

\newcommand{\si}{\mathscr{I}}

\newcommand{\usi}{\underline{\mathscr{I}}}
\newcommand{\usj}{\underline{\mathscr{J}}}
\newcommand{\usk}{\underline{\mathscr{K}}}

\newcommand{\hyp}[2]{\mathbb{h}^{#1}_{#2}}

\newcommand{\ssacone}[1]{\text{\raisebox{-.09cm}{\Large{$\boxtimes$}}$^{\text{SSA}}_{#1}$}}

\newcommand{\pmi}{\mathscr{P}}
\newcommand{\pmit}{\mathscr{P}}
\definecolor{maxcolor}{rgb}{1,0.03,0}
\definecolor{dblue}{rgb}{0.25,0.03,0.8}
\definecolor{vecolor}{rgb}{0.7,0.3,0.9}
\definecolor{sergiocolor}{rgb}{0.1,0.7,0.1}

\begin{document}

\title{The quantum marginal independence problem}

\author{Sergio Hern\'andez Cuenca}
\email{sergiohc@physics.ucsb.edu}
\affiliation{Department of Physics, University of California, Santa Barbara, CA 93106, USA}

\author{Veronika E. Hubeny}
\email{veronika@physics.ucdavis.edu}
\affiliation{Center for Quantum Mathematics and Physics (QMAP),
Department of Physics, University of California, Davis, CA 95616 USA}

\author{Mukund Rangamani}
\email{mukund@physics.ucdavis.edu}
\affiliation{Center for Quantum Mathematics and Physics (QMAP),
Department of Physics, University of California, Davis, CA 95616 USA}

\author{Massimiliano Rota}
\email{m.rota@uva.nl}
\affiliation{Institute for Theoretical Physics, University of Amsterdam,
Science Park 904, Postbus 94485, 1090 GL Amsterdam, The Netherlands}
\affiliation{Department of Physics, University of California, Santa Barbara, CA 93106, USA}

\begin{abstract}
{
We investigate whether the presence or absence of correlations between subsystems of an $\N$-partite quantum system is solely constrained by the non-negativity and monotonicity of mutual information. We argue that this relatively simple question is in fact very deep  because it  is sensitive to the structure of the set of $\N$-partite states. It can be informed by  inequalities satisfied by the von Neumann entropy, but has the advantage of being more tractable. We exemplify this by deriving the explicit solution for $\N=4$, despite having limited knowledge of the entropic inequalities.  Furthermore, we describe how this question can be tailored to the analysis of more specialized classes of states such as classical probability distributions, stabilizer states, and geometric states in the holographic gauge/gravity duality.
 }
\end{abstract}

\maketitle

%~~~~~~~~~~~~~~~~~~~~~~~~~~~~~~~~~~~~~~~~~~~~~~~
\section{Introduction}
\label{sec:intro}
%~~~~~~~~~~~~~~~~~~~~~~~~~~~~~~~~~~~~~~~~~~~~~~

Understanding the structure of the set of $\N$-partite quantum states is of paramount importance for both quantum information theory \cite{Nielsen:2011aa,Wilde:2013aa} and fundamental physics. To attack this problem, one typically considers a measure of correlations between subsystems and analyzes constraints, usually in the form of inequalities, on the possible values of this measure. This is for example the approach followed by  \cite{Pippenger:2003aa,Linden_2005,Cadney:2011aa} for the von Neumann entropy, by \cite{Linden_2013} for R\'enyi entropies, and by \cite{Ibinson_2006} for relative entropy.

In this letter we introduce a novel and simpler version of this problem, independent from \textit{quantification} of correlations, which leads to constraints of a different type.

Let us first make a simple observation about an intuitive property of correlations: a quantum system $\a$ correlated with another system $\b$ is necessarily correlated with any composite system $\b\c$. Similarly, if $\a$ is uncorrelated with $\b\c$, it cannot be correlated with either of the individual subsystems $\b$ or $\c$. In quantum mechanics these statements are an immediate consequence of the monotonicity of mutual information \cite{Nielsen:2011aa}. However, they illustrate a fundamental fact: in a multipartite setting the presence or absence of correlations between some subsystems imposes constraints on the allowed correlations of other subsystems.

Motivated by this, we pose a general question: \emph{are non-negativity and monotonicity of mutual information the only constraints to the presence or absence of correlations between subsystems in an $\N$-partite quantum system}? 

We demonstrate that, while simply phrased, this \emph{quantum marginal independence problem (QMIP)} is sensitive to the structure of the set of $\N$-partite quantum states, and that its relative simplicity could allow for a general solution for arbitrary $\N$. To show this, we will use known results about entropy inequalities, however, the formalization of the problem will make clear that the constraints on this structure are naturally expressed in the language of partially ordered sets and lattices.\footnote{For related work in this direction see 
\cite{Fritz:2012en,Szalay:2018aa}.}

Although the question is formalized here for general multipartite quantum systems, the framework is naturally adaptable to more restricted sets of states. One example is the case of stabilizer states \cite{Gottesman:1997aa,Klappenecker:2002aa,Klappenecker:2002ab} (SMIP), which have a broad spectrum of applications in quantum information and computing. Another interesting example appears in quantum gravity. 

In the context of the gauge/gravity duality we can consider states which are `geometric'. For these, the mutual information has been shown to be monogamous (MMI) \cite{Hayden:2011ag}, and other inequalities have been derived up to $\N=5$ \cite{Bao:2015bfa,Cuenca:2019uzx}. Unlike the more general quantum case (where the solution to the QMIP can at best contain partial information about the existence or the structure of new entropy inequalities), in the holographic context there are reasons to expect that the solution to the holographic marginal independence problem (HMIP) provides a key ingredient for the systematic derivation of these inequalities \cite{Hubeny_2018,Hubeny_2019,Hernndez-Cuenca:aa}, and elucidate the holographic dictionary in quantum gravity.

%~~~~~~~~~~~~~~~~~~~~~~~~~~~~~~~~~~~~~~~~~~~~~~~
\section{Definition of the QMIP}
\label{sec:qmip}
%~~~~~~~~~~~~~~~~~~~~~~~~~~~~~~~~~~~~~~~~~~~~~~

Consider an arbitrary $\N$-partite quantum system described by a density matrix $\rho$ acting on a Hilbert space
\begin{equation}
\mathcal{H}=\mathcal{H}_1\otimes\mathcal{H}_2\otimes\cdots\otimes\mathcal{H}_\N\;.
\label{eq:hilbert}
\end{equation}
A non-empty \textit{subsystem} is associated to a collection of factors in \eqref{eq:hilbert}, and is labeled by an index $\si\in\wp_*([\N])$, where $\wp_*([\N])=\wp([\N])\!\smallsetminus\!\emptyset$, and $\wp([\N])$ is the power set of 
$[\N]=\{1,2,\ldots,\N\}$.

Given $\rho$, the von Neumann entropies of the marginals of all subsystems can be arranged into an \textit{entropy vector}
\begin{equation}
\mathbf{S}(\rho)=\{S_\si \in \mathbb{R}_{\geq 0} \,|\, \si\in\wp_*([\N])\}\;.
\end{equation}
These vectors live in a $\D$-dimensional vector space $\mathbb{R}^\D$, with $\D=2^\N-1$,  called \textit{entropy space}.

It will be convenient to consider a purification of $\rho$ by an ancillary subsystem on $\mathcal{H}_{\N+1}$. Consequently, we introduce an \textit{extended} index $\usi\in\wp_*([\N+1])$ to label subsystems which can include the ancilla.

For two non-overlapping subsystems $\usi$ and $\usk$, i.e., $\usi\cap\usk=\emptyset$, their mutual information is defined as 
\begin{equation}
\I_2(\usi:\usk)= S_{\usi}+S_{\usk}-S_{\usi\cup\usk}\;.
\label{eq:mutual_info}
\end{equation}
We will call an expression like \eqref{eq:mutual_info}, for a specific choice of indices in an $\N$-partite setting, an \textit{instance} of the mutual information. When \eqref{eq:mutual_info} vanishes, the two subsystems $\usi$ and $\usk$ are \textit{independent} \cite{Nielsen:2011aa}.

Geometrically, we can think of the equation
\begin{equation}
\hyp{\usi}{\usk}:\quad\I_2(\usi:\usk)=0
\label{eq:hyperplane}
\end{equation}
as representing a hyperplane $\hyp{\usi}{\usk}$ in entropy space.\footnote{For any index $\usj$ which includes the ancilla we replace $S_{\usj}\rightarrow S_{\usj^c}$.} If for a given $\rho$, two subsystems $\usi$ and $\usk$ are independent, the entropy vector $\mathbf{S}(\rho)$ belongs to this hyperplane. More generally, if different instances of the mutual information vanish, then multiple subsystems are pairwise independent and the entropy vector $\mathbf{S}(\rho)$ belongs to the \textit{intersection} of the corresponding hyperplanes. 

The analysis of bipartite independence for the marginals of a multipartite density matrix is then naturally linked to the analysis of how different hyperplanes $\hyp{\usi}{\usk}$ intersect, leading to the following definition:

\begin{definition}
\emph{\bf{(Mutual information arrangement)}} The $\N$-\emph{party mutual information arrangement (MIA$_\N$)} is a hyperplane arrangement in $\N$-party entropy space consisting of the set of all hyperplanes of the form \eqref{eq:hyperplane} for all instances \eqref{eq:mutual_info} of the mutual information. Formally, the \emph{MIA}$_\N$ is defined by the following equation
\begin{equation}
 \prod_{\substack{\usi\cap\usk=\emptyset\\ \usi\cup\usk\neq [\N+1]}}\I_2(\usi:\usk)=0\;. 
\end{equation}
\end{definition}

The MIA is a finite arrangement of cardinality
 $3\genfrac\{\}{0pt}{1}{\N+1}{3}\sim(3/2)^{\N+2}$ \cite{Hubeny_2019}, and it is invariant under any permutation of the $\N$ fundamental subsystems and the ancilla (cf. \cite{Stanley:2004aa,Orlik:1992aa} for hyperplane arrangements).

Since the MIA is a \textit{central} arrangement (the intersection of all the hyperplanes contains the origin), the intersection of an arbitrary collection of hyperplanes is a \textit{linear subspace} of entropy space. It will be convenient to distinguish a subspace $x$ from its \textit{variety} $\mathscr{V}(x)$, which comprises of the set of points belonging to it. 

The set of all these subspaces can be made into a partially ordered set by introducing an order relation, customarily chosen as reverse inclusion (as it gives a richer structure to the poset of central arrangements \cite{Orlik:1992aa}), viz., 
\begin{equation}
x\preceq y\iff \mathscr{V}(x)\supseteq \mathscr{V}(y)\;.
\label{eq:order}
\end{equation}

Moreover, the intersection poset has the structure of a \textit{lattice}\footnote{A \textit{lattice} is a poset such that any subset has a supremum and an infimum \cite{Birkhoff:1940aa}.} \cite{Orlik:1992aa}. Notice in particular that the full entropy space is an element of this poset. We will refer to this lattice as the MIA-\textit{lattice} and denote it by $\mathbb{L}(\text{MIA})$. The elements of the MIA-lattice are precisely what we need to characterize the bipartite independence of subsystems for a given density matrix:

\begin{definition}
\emph{\bf{(Pattern of marginal independence)}} A \emph{pattern of marginal independence} in an $\N$-partite setting is a linear subspace $\mathbb{s}$ of entropy space which is an element of $\mathbb{L}(\text{\emph{MIA}}_\N)$. 
\end{definition}

To any vector $\mathbf{S}$ of entropy space we can associate a pattern of marginal independence as follows
\begin{equation}
\pmi(\mathbf{S})=\sup\;\{\mathbb{s}\in\mathbb{L}(\text{MIA}_\N)\,|\, \mathbf{S}\in\mathscr{V}(\mathbb{s})\}\;.
\label{eq:pmi}
\end{equation}
The map \eqref{eq:pmi} is well defined since the supremum is guaranteed to exist\footnote{ Since the entire entropy space is an element of $\mathbb{L}(\text{MIA})$ the map $\pmi(\mathbf{S})$ exists for any $\mathbf{S}$.} (because $\mathbb{L}(\text{MIA}_\N)$ is a lattice), and it is unique. Equivalently, one can think of the subspace $\pmi(\mathbf{S})$ as the lowest-dimension subspace in the MIA-lattice which contains $\mathbf{S}$. We can then use this map to associate a pattern of marginal independence $\pmi(\mathbf{S}(\rho))$ to any density matrix $\rho$.

Conversely, given a pattern of marginal independence, we can ask if there exists an $\N$-partite density matrix $\rho$ which realizes it.
\begin{definition}
\emph{\bf{(Realizable pattern)}} A pattern of marginal independence $\mathbb{s}\in\mathbb{L}(\text{\emph{MIA}}_\N)$ is \emph{realizable} (or, equivalently, the corresponding subspace is \emph{accessible}), if there exists an $\N$-partite Hilbert space \eqref{eq:hilbert}, and a density matrix $\rho$, such that $\mathbb{s}=\pmi(\mathbf{S}(\rho))$. 
\end{definition}

With these definitions at hand, we can finally introduce the formal statement of the \textit{quantum marginal independence problem} (QMIP)
\begin{problem}
\emph{\bf{(Quantum marginal independence)}} For any given $\N$, what are all the realizable patterns of marginal independence?
\end{problem}

While we have defined the problem for arbitrary quantum states, it is clear that one can consider more specialized versions of it by restricting to particular classes of states. For example, one can define the same problem for stabilizer states (SMIP), classical probability distributions (CMIP) or geometric states in holography (HMIP). 

In the general case, for arbitrary $\N$, not all subspaces of the MIA-lattice are accessible as a consequence of universal inequalities satisfied by the mutual information. For restricted sets of states, additional inequalities could hold. In \S\ref{sec:conj} we explain how to derive the proper subset of the MIA-lattice which is `compatible' with a given set of inequalities.

%~~~~~~~~~~~~~~~~~~~~~~~~~~~~~~~~~~~~~~~~~~~~~~~
\section{Constraining the QMIP by entropy inequalities}
\label{sec:conj}
%~~~~~~~~~~~~~~~~~~~~~~~~~~~~~~~~~~~~~~~~~~~~~~

It is well known that for any density matrix and any choice of subsystems, the mutual information satisfies subadditivity (SA) and strong subadditivity (SSA)  \cite{Nielsen:2011aa},
\begin{subequations}
\begin{align}
&\text{SA}: \quad \;\; \I_2(\usi:\usk)\geq 0 \,,
\label{eq:sa}\\
&\text{SSA}: \quad \I_2(\usi:\usj\cup\usk)-\I_2(\usi:\usk)\geq 0\,.
\label{eq:ssa}
\end{align}
\label{eq:sassa}
\end{subequations}

Because of these constraints, some patterns of marginal independence are certainly \textit{non-realizable}. We now determine what subset of patterns of the MIA-lattice is actually relevant to the QMIP. We begin by defining the following object in entropy space:

\begin{definition}
\emph{\bf{(SSA polyhedron)}} The $\N$-party \emph{strong-subadditivity polyhedron}, denoted by $\emph{\ssacone{\N}}$, is the intersection of the half-spaces associated to all possible instances of the inequalities \eqref{eq:sassa}. 
\label{def:polyhedron}
\end{definition}

For any polyhedron, a \textit{face} $\mathbb{f}$ is defined as the intersection of the polyhedron with a hyperplane bounding a half-space which contains the polyhedron. The set of faces of a polyhedron again forms a lattice, with the entire polyhedron also being included as a face \cite{Ziegler:1995aa}. We refer to the lattice of faces of the SSA polyhedron as the SSA-\textit{lattice}, denoted by $\mathbb{L}(\ssacone{\N})$. Like for the MIA-lattice, we distinguish between an element $\mathbb{f}$ of $\mathbb{L}(\ssacone{\N})$ and its variety  $\mathscr{V}(\mathbb{f})$. 

We define the \textit{interior} of a face $\mathbb{f}$ as the \textit{variety} 
\begin{equation}
\text{int}\;\mathbb{f}=\mathscr{V}(\mathbb{f})\smallsetminus\partial\mathscr{V}(\mathbb{f}),
\label{eq:interior}
\end{equation}
taking $\text{int}\;\mathbf{0}=\mathbf{0}$ for the origin.  It then follows that for any face, all the points in the interior are mapped by \eqref{eq:pmi} to the same element of the MIA-lattice:
\begin{lemma}
For any $\N$, and for any face $\mathbb{f}\in\mathbb{L}(\emph{\ssacone{\N}})$,
\begin{equation}
 \forall\,\mathbf{S}\in\text{\emph{int}}\,\mathbb{f} \,, \quad \pmit(\mathbf{S})=\mathbb{s}\;\; {\rm with}\;\, \mathbb{s} \in \mathbb{L}(\rm{MIA}_\N).
\end{equation}
\label{lem:lemma1}
\end{lemma}
\begin{proof}
Consider a face $\mathbb{f}$ and a vector $\mathbf{S}^*\in\text{int}\,\mathbb{f}$, and let $\mathbb{s}^*=\pmi(\mathbf{S}^*)$. If for another vector $\mathbf{S}\in\mathbb{f}$ one has $\pmi(\mathbf{S}) = \mathbb{s}\neq\mathbb{s}^*$, then $\mathbb{s} \succ \mathbb{s}^*$. Thus, there exists a hyperplane $\mathbb{h}\in\mathbb{L}(\text{MIA}_\N)$ such that $\mathbf{S}\in\mathscr{V}(\mathbb{h})$ but $\mathbf{S}^*\notin\mathscr{V}(\mathbb{h})$. By Def.~\ref{def:polyhedron}, any $\mathbb{h}\in\mathbb{L}(\text{MIA}_\N)$ is the boundary of a subspace which includes $\ssacone{\N}$, therefore its intersection with the polyhedron is a face $\mathbb{f}'\neq\mathbb{f}$. This implies that $\mathbf{S}\notin\text{int}\;\mathbb{f}$. 
\end{proof}

We denote the element of $\mathbb{L}(\text{MIA})$ which is the image of $\text{int}\,\mathbb{f}$ under \eqref{eq:pmi} by $\pmit(\mathbb{f})$. Under this mapping, different faces can have the same image; this fact will play an important role in \S\ref{sec:cases}.
We are finally ready to define the construct of interest, namely the SSA-compatible subset of the MIA-lattice
\begin{equation}
\mathbb{G}_\N=\pmit\left(\mathbb{L}(\ssacone{\N})\right).
\label{eq:Gdef}
\end{equation}

The reason for focusing on this subset owes to 
\begin{lemma}
For any $\N$, a pattern of marginal independence $\mathbb{s}\in\mathbb{L}(\text{\emph{MIA}}_\N)$ which is not in $\mathbb{G}_\N$ is non-realizable. Furthermore, for any $\mathbb{s}\in\mathbb{G}_\N$, there exists at least one entropy vector $\mathbf{S}\in\mathbb{R}^\D$ with  $\pmit(\mathbf{S})=\mathbb{s}$ which satisfies all instances of \eqref{eq:sassa}.
\end{lemma}
\begin{proof}
i) From \eqref{eq:interior} and basic properties of polyhedra \cite{Ziegler:1995aa}, it follows that any vector $\mathbf{S}\in\mathbb{R}^\D$ which satisfies all instances of \eqref{eq:sassa} is an element of the interior of at least one face of the SSA polyhedron. Its image under \eqref{eq:pmi} is therefore in $\mathbb{G}_\N$. Hence if a vector $\mathbf{S}\in\mathbb{R}^\D$ is mapped to an element $\mathbb{s}^*\in\mathbb{L}(\text{MIA}_\N)$ which is not in $\mathbb{G}_\N$ it must violate at least one instance of \eqref{eq:sassa}. Therefore for any Hilbert space \eqref{eq:hilbert} there is no density matrix such that $\pmi(\mathbf{S}(\rho))=\mathbb{s}^*$.

ii) Consider an element $\mathbb{s}\in\mathbb{G}_\N$. By Lemma \ref{lem:lemma1} and definition \eqref{eq:Gdef}, there exists a face $\mathbb{f}\in\ssacone{\N}$ such that $\mathscr{P}(\mathbf{S})=\mathbb{s}$ for all $\mathbf{S}\in\text{int}\;\mathbb{f}$. The result follows from the fact that $\text{int}\;\mathbb{f}\neq\emptyset$.
\end{proof}

Therefore, $\mathbb{G}_\N$ contains exactly all the patterns of marginal independence which could possibly be realized in accordance to SA and SSA. While this does not imply that each pattern in $\mathbb{G}_\N$ actually can be realized, it allows us to state an upgraded version of the QMIP:

\begin{problem}
For any given $\N$, which patterns of marginal independence in $\mathbb{G}_\N$ are realizable?
\end{problem}

As mentioned previously, one might be interested in restricted versions of the QMIP. If the states of interest are known to satisfy a collection $\{\text{Q}\}$ of additional entropy inequalities, the construction presented here can readily be adapted, since it did not depend on any structural detail of \eqref{eq:sassa} other than linearity. One first defines a new Q-polyhedron by adding the list $\{\text{Q}\}$ to Def.~\ref{def:polyhedron}, and then proceeds to define a new subset $\mathbb{G}_{\N}^\text{Q}$ of patterns that could be realized by states in the allowed class. We employ this strategy to completely solve the QMIP for $\N=4$ and analyze the more involved $\N=5$ case in \S\ref{sec:cases}.

%~~~~~~~~~~~~~~~~~~~~~~~~~~~~~~~~~~~~~~~~~~~~~~~
\section{Solution to the QMIP for $\N\leq 5$}
\label{sec:cases}
%~~~~~~~~~~~~~~~~~~~~~~~~~~~~~~~~~~~~~~~~~~~~~~

We now present the solution to the QMIP for $\N\leq 4$, along with some observations about the $\N=5$ case, and its relation to the restricted problem for stabilizer states (SMIP).

\medskip 
\noindent
{\bf $\N=2$:}
This case is special since the set of inequalities \eqref{eq:sassa} reduces to instances of SA only and defines a simpler object: the SA-polyhedron. The QMIP is unconstrained, in the sense that $\mathbb{G}_2= \mathbb{L}(\text{MIA}_2)$. It is straightforward to check that every element of $\mathbb{G}_2$ is realizable.

\medskip 
\noindent
{\bf $\N=3$:}
In this case, entropy inequalities do constrain the set of possibly realizable patterns, $\mathbb{G}_3\subset \mathbb{L}(\text{MIA}_3)$. The question then is whether each pattern in $\mathbb{G}_3$ is realizable. To check that this is the case, we first find the extreme rays of the SSA polyhedron,\footnote{All the polyhedra we work with are in fact polyhedral cones.} and then construct quantum states with entropy vectors corresponding to these rays. These are Bell pairs, the GHZ state of four qubits, and a four-partite perfect state \citep{Pippenger:2003aa}. Taking appropriate tensor products, one can explicitly construct quantum states realizing all patterns in $\mathbb{G}_3$.\footnote{If $\rho_1,\rho_2$ realize the patterns $\mathbb{s}_1,\mathbb{s}_2$, then $\rho_1\otimes\rho_2$ realizes the pattern $\mathbb{s}=\mathbb{s}_1+\mathbb{s}_2$, where $\mathscr{V}(\mathbb{s})=\mathscr{V}(\mathbb{s}_1)+\mathscr{V}(\mathbb{s}_2)$ and ``$+$'' is the usual (non-direct) sum of linear subspaces.}

\medskip 
\noindent
{\bf $\N=4$:} Density matrices for $4$ parties saturating certain instances of SSA satisfy new inequalities which are independent from \eqref{eq:sassa} \cite{Linden_2005}. It is, however, unknown if unconstrained inequalities exist \cite{Cadney:2011aa}. Interestingly, the QMIP nevertheless turns out to be constrained by \eqref{eq:sassa} alone.

To see this, consider the restricted problem of focusing on stabilizer states \cite{Klappenecker:2002aa,Klappenecker:2002ab}. For $4$ parties, the entropies of stabilizer states satisfy the Ingleton inequality \cite{Linden:2013aa}. We include this additional constraint in defining the SMIP as indicated in \S\ref{sec:conj}. From the `Ingleton polyhedron' obtained by adding this inequality to the list in Def.~\ref{def:polyhedron}, we find the $\mathbb{G}_4^\text{ING}\subset\mathbb{L}(\text{MIA}_4)$, the subset of patterns realizable by stabilizer states. Extreme rays of this polyhedron have been explicitly realized by quantum states \cite{Linden:2013aa}. Therefore, the patterns in $\mathbb{G}_4^\text{ING}$ may be realized taking appropriate tensor products of these states. Constructing the usual set $\mathbb{G}_4\subset\mathbb{L}(\text{MIA}_4)$ from the SSA polyhedron we then check that $\mathbb{G}_4=\mathbb{G}_4^\text{ING}$.

\medskip 
\noindent
{\bf $\N=5$:} Little is known about the $5$-partite stabilizer cone (let alone the quantum entropy cone). Nevertheless, we motivate the need for new constraints. In the preceding examples, all patterns in $\mathbb{G}_{\N}$ were realizable, in particular already within the class of stabilizer states. This ceases to be true for $\N=5$.

Any entropy vector realized by a $5$-party stabilizer state has to satisfy the Ingleton inequality. Uplifting the $\N=4$ Ingleton polyhedron supplemented with all instances of the Ingleton inequality for $\N=5$, we build $\mathbb{G}^\text{ING}_5$. One then checks that $\mathbb{G}_5\supset\mathbb{G}_5^\text{ING}$. This result suggests two interesting potential scenarios that we discuss in \S\ref{sec:discuss}. Furthermore, one can check that some patterns in $\mathbb{G}_5\smallsetminus\mathbb{G}_5^\text{ING}$ can only potentially be realized by density matrices which violate monotonicity of the entropy, and therefore not by classical probability distributions.

%~~~~~~~~~~~~~~~~~~~~~~~~~~~~~~~~~~~~~~~~~~~~~~~
\section{Discussion}
\label{sec:discuss}
%~~~~~~~~~~~~~~~~~~~~~~~~~~~~~~~~~~~~~~~~~~~~~~

\begin{table}[tb]
	\setlength{\tabcolsep}{6pt}
	\begin{tabular}{|m{6pt}|c|c|c|}
		\hline\vspace{1pt}
		$\N\;\quad$ & QMIP  & SMIP & HMIP \\ \hline\vspace{1.5pt}
		\centering2 & SA &  SA & SA \\
		\centering3 & SSA &  SSA & SSA \\
		\centering4 & SSA & SSA & MMI \\
		\centering5 & SSA (?) &  ING (?) & MMI \\
		\hline
	\end{tabular}
	\caption{Solutions to the Quantum (Q), Stabilizer (S) and Holographic (H) MIP for $\N\in\{2,3,4,5\}$.
		Question marks denote conjectures.}
	\label{tab:lowNMIPs}
\end{table}

Table~\ref{tab:lowNMIPs} summarizes the solutions to the various MIPs we considered. For each case, the solution is labeled by the weakest inequality Q which should be added to Def.~\ref{def:polyhedron} in order to obtain a set of patterns $\mathbb{G}^\text{Q}_\N$ which are all realizable by states in the corresponding class.

The holographic case (HMIP) will be discussed extensively in \citep{Hernndez-Cuenca:aa}, but we summarize here the most salient aspects of the solution for $\N\leq5$. While for $\N=3$ geometric states are constrained by MMI \cite{Hayden:2011ag}, the HMIP is insensitive to it, since geometric states allow to realize all possible patterns of marginal independence which are consistent with quantum mechanics. Interestingly, this inequality becomes instead crucial for $\N=4$, and it remains the only relevant constraint for the solution to this problem even for $\N=5$, despite the presence of additional inequalities \cite{Bao:2015bfa,Cuenca:2019uzx}.

Another specialized version of the problem that we did not discuss pertains to classical probabilities distributions (CMIP). In this case, an interesting question is whether \textit{Shannon inequalities} are the only necessary constraints, or if instead additional inequalities impose further restrictions \cite{Zhang:1998aa,Matus:2007aa,K.-Makarychev:2002aa}. The question is also relevant for the analysis of the SMIP, since stabilizer states satisfy all balanced classical inequalities \cite{Gross:2013aa}.

Regarding the SMIP, and its relation to the general QMIP, our analysis of the $\N=5$ case suggests two potential general scenarios which deserve further investigation. One possibility is that the full solution to the QMIP is given by $\mathbb{G}_\N$ itself, with non-stabilizer states being essential for the realization of certain patterns. Another is that stabilizer states suffice to realize all patterns allowed in quantum mechanics, and since they realize a proper subset of $\mathbb{G}_\N$, there must exist new entropy inequalities.

Finally, let us comment on the general logic of our approach and how it should be extended. While we have formulated the QMIP in the language of ordered sets, we have used the varieties associated to the elements of these sets to relate the QMIP to known results about inequalities for the von Neumann entropy. This strategy was sufficient for the scope of this letter, and in particular to solve the $\N=4$ case. However, in order to use the QMIP to extract new structural properties of the set of $\N$-partite quantum states, one will need to understand how to solve the QMIP \textit{without} relying on such inequalities. In particular, this should make it possible to formulate the constraints which determine the subset of realizable patterns of the MIA-lattice purely in the language of order theory and combinatorics.

%~~~~~~~~~~~~~~~~~~~~~~~~~~~~~~~~~~~~~~~~~~~~~~
\acknowledgments 
%~~~~~~~~~~~~~~~~~~~~~~~~~~~~~~~~~~~~~~~~~~~~~~
It is a pleasure to thank Temple He, Christian Majenz and Michael Walter for discussions.
S.~Hern\'andez Cuenca was supported by fellowship LCF/BQ/AA17/11610002 from ``la Caixa'' Foundation (ID 100010434) and by NSF grant PHY-1801805.
V.~Hubeny and M.~Rangamani were supported by U.S.\ Department of Energy grant {DE-SC0019480} under the HEP-QIS QuantISED program and by funds from the University of California.
M.~Rota was supported by the Simons Foundation, under the ``It From Qubit'' collaboration, by funds from the University of California, and by the University of Amsterdam, via the ERC Consolidator Grant QUANTIVIOL.

% \bibliographystyle{utcaps}
% %\bibliographystyle{apsrev}
% \bibliography{references}

\begin{thebibliography}{10}

\bibitem{Nielsen:2011aa}
M.~Nielsen and I.~Chuang, {\em Quantum Computation and Quantum Information}.
\newblock Cambridge University Press, 2011.

\bibitem{Wilde:2013aa}
M.~M. Wilde, {\em Quantum Information Theory}.
\newblock Cambridge University Press, 2013.

\bibitem{Pippenger:2003aa}
N.~Pippenger, ``The inequalities of quantum information theory,''
  \href{http://dx.doi.org/10.1109/TIT.2003.809569}{{\em IEEE Trans. Inf.
  Theory} {\bfseries 49} no.~4, (2003) 773--789}.

\bibitem{Linden_2005}
N.~Linden and A.~Winter, ``A New Inequality for the von Neumann Entropy,''
  \href{http://dx.doi.org/10.1007/s00220-005-1361-2}{{\em Comm Math Phys}
  {\bfseries 259} no.~1, (May, 2005) 129--138},
  \href{http://arxiv.org/abs/0406162}{{\ttfamily arXiv:0406162 [quant-ph]}}.

\bibitem{Cadney:2011aa}
J.~Cadney, N.~Linden, and A.~Winter, ``Infinitely many constrained inequalities
  for the von Neumann entropy,''
  \href{http://dx.doi.org/10.1109/TIT.2012.2185036}{{\em IEEE Trans. Inf.
  Theory} {\bfseries 58} no.~6, (2012) 3657--3663},
  \href{http://arxiv.org/abs/1107.0624}{{\ttfamily arXiv:1107.0624
  [quant-ph]}}.

\bibitem{Linden_2013}
N.~Linden, M.~Mosonyi, and A.~Winter, ``The structure of R{\'e}nyi entropic
  inequalities,'' \href{http://dx.doi.org/10.1098/rspa.2012.0737}{{\em P Roy
  Soc A-Math Phy} {\bfseries 469} no.~2158, (Oct, 2013) 20120737},
  \href{http://arxiv.org/abs/1212.0248}{{\ttfamily arXiv:1212.0248
  [quant-ph]}}.

\bibitem{Ibinson_2006}
B.~Ibinson, N.~Linden, and A.~Winter, ``All Inequalities for the Relative
  Entropy,'' \href{http://dx.doi.org/10.1007/s00220-006-0081-6}{{\em
 Comm Math Phys} {\bfseries 269} no.~1, (Aug, 2006)
  223--238}, \href{http://arxiv.org/abs/0511260}{{\ttfamily arXiv:0511260
  [quant-ph]}}.

\bibitem{Fritz:2012en}
T.~Fritz and R.~Chaves, ``Entropic inequalities and marginal problems,'' \href{https://dx.doi.org/10.1109/TIT.2012.2222863}{{\em
  IEEE Trans. Inf. Theory} {\bfseries 59} no.~2, (2012) 803--817},
\href{https://arxiv.org/abs/1112.4788}{{\ttfamily arXiv:1112.4788
[cs.IT]}}.

\bibitem{Szalay:2018aa}
S.~Szalay, ``The classification of multipartite quantum correlation,''
  \href{http://dx.doi.org/10.1088/1751-8121/aae971}{{\em J Phys A: Math Theor}
  {\bfseries 51} no.~48, (Nov, 2018) 485302},
  \href{http://arxiv.org/abs/1806.04392}{{\ttfamily arXiv:1806.04392
  [quant-ph]}}.

\bibitem{Gottesman:1997aa}
D.~Gottesman, {\em Stabilizer Codes and Quantum Error Correction}.
\newblock PhD thesis, California Institute of Technology, 1997, 
\newblock \href{http://arxiv.org/abs/9705052}{{\ttfamily arXiv:9705052
  [quant-ph]}}.

\bibitem{Klappenecker:2002aa}
A.~Klappenecker and M.~Roetteler, ``Beyond Stabilizer Codes I: Nice Error
  Bases,'' \href{http://dx.doi.org/10.1109/TIT.2002.800471}{{\em IEEE Trans.
  Inf. Theory} {\bfseries 48} no.~8, (2002) 2392--2395},
  \href{http://arxiv.org/abs/0010082}{{\ttfamily arXiv:0010082 [quant-ph]}}.

\bibitem{Klappenecker:2002ab}
A.~Klappenecker and M.~Roetteler, ``Beyond Stabilizer Codes II: Cliord Codes,''
  \href{http://dx.doi.org/10.1109/TIT/2002.800473}{{\em IEEE Trans. Inf.
  Theory} {\bfseries 48} no.~8, (2002) 2396--2399},
  \href{http://arxiv.org/abs/0010076}{{\ttfamily arXiv:0010076 [quant-ph]}}.

\bibitem{Hayden:2011ag}
P.~Hayden, M.~Headrick, and A.~Maloney, ``{Holographic Mutual Information is
  Monogamous},'' \href{http://dx.doi.org/10.1103/PhysRevD.87.046003}{{\em Phys.
  Rev.} {\bfseries D87} no.~4, (2013) 046003},
\href{http://arxiv.org/abs/1107.2940}{{\ttfamily arXiv:1107.2940 [hep-th]}}.
%%CITATION = ARXIV:1107.2940;%%.

\bibitem{Bao:2015bfa}
N.~Bao, S.~Nezami, H.~Ooguri, B.~Stoica, J.~Sully, and M.~Walter, ``{The
  Holographic Entropy Cone},''
  \href{http://dx.doi.org/10.1007/JHEP09(2015)130}{{\em JHEP} {\bfseries 09}
  (2015) 130},
\href{http://arxiv.org/abs/1505.07839}{{\ttfamily arXiv:1505.07839 [hep-th]}}.
%%CITATION = ARXIV:1505.07839;%%.

\bibitem{Cuenca:2019uzx}
S.~Hern\'andez~Cuenca, ``{Holographic entropy cone for five regions},''
  \href{http://dx.doi.org/10.1103/PhysRevD.100.026004}{{\em Phys. Rev.}
  {\bfseries D100} no.~2, (2019) 026004},
\href{http://arxiv.org/abs/1903.09148}{{\ttfamily arXiv:1903.09148 [hep-th]}}.
%%CITATION = ARXIV:1903.09148;%%.

\bibitem{Hubeny_2018}
V.~E. Hubeny, M.~Rangamani, and M.~Rota, ``Holographic Entropy Relations,''
  \href{http://dx.doi.org/10.1002/prop.201800067}{{\em Fortschr Physik}
  {\bfseries 66} no.~11-12, (Sep, 2018) 1800067},
  \href{http://arxiv.org/abs/1808.07871}{{\ttfamily arXiv:1808.07871
  [quant-ph]}}.

\bibitem{Hubeny_2019}
V.~E. Hubeny, M.~Rangamani, and M.~Rota, ``The Holographic Entropy
  Arrangement,'' \href{http://dx.doi.org/10.1002/prop.201900011}{{\em Fortschr
  Physik} {\bfseries 67} no.~4, (Feb, 2019) 1900011},
  \href{http://arxiv.org/abs/1812.08133}{{\ttfamily arXiv:1812.08133
  [hep-th]}}.

\bibitem{Hernndez-Cuenca:aa}
S.~Hern\'andez~Cuenca, V.~E. Hubeny, M.~Rangamani, and M.~Rota, ``The
  holograhic marginal independence problem,'' {\em To appear} .

\bibitem{Stanley:2004aa}
R.~P. Stanley, ``An introduction to hyperplane arrangements,'' {\em Geometric
  combinatorics} {\bfseries 13} (2004) 389--496.

\bibitem{Orlik:1992aa}
P.~Orlik and H.~Terao, {\em Arrangements of Hyperplanes}.
\newblock Springer-Verlag, 1992.

\bibitem{Birkhoff:1940aa}
G.~Birkhoff, {\em Lattice Theory}.
\newblock American Mathematical Society, 3rd edition, 1940.

\bibitem{Ziegler:1995aa}
G.~M. Ziegler, {\em Lectures on Polytopes}.
\newblock Springer-Verlag, 1995.

\bibitem{Linden:2013aa}
N.~Linden, F.~Mat{\'u}{\v s}, M.~B. Ruskai, and A.~Winter, ``The Quantum
  Entropy Cone of Stabiliser States,''
  \href{http://dx.doi.org/10.4230/LIPIcs.TQC.2013.270}{{\em Proc. 8th TQC
  Guelph, LIPICS} {\bfseries 22} (2013) 270--284},
  \href{http://arxiv.org/abs/1302.5453}{{\ttfamily arXiv:1302.5453
  [quant-ph]}}.

\bibitem{Zhang:1998aa}
Z.~Zhang and R.~W. Yeung, ``On characterization of entropy function via
  information inequalities,'' \href{https://dx.doi.org/10.1109/18.681320}{{\em IEEE Trans. Inf. Theory} {\bfseries 44} no.~4, (1998) 1440--1452}.

\bibitem{Matus:2007aa}
F.~Mat{\'u}{\v s}, ``Infinitely Many Information Inequalities,'' \href{https://dx.doi.org/10.1109/ISIT.2007.4557201}{{\em Proc.
  ISIT} (2007) 41--44}.

\bibitem{K.-Makarychev:2002aa}
K.~Makarychev, Y.~Makarychev, A.~Romashchenko, and N.~Vereshchagin, ``A new
  class of non-Shannon type inequalities for entropies,''\href{http://dx.doi.org/10.4310/CIS.2002.v2.n2.a3}{{\em Commun. Inf.
  Syst.} {\bfseries 2} no.~2, (2002) 147--166}.

\bibitem{Gross:2013aa}
D.~Gross and M.~Walter, ``Stabilizer information inequalities from phase space
  distributions,'' \href{http://dx.doi.org/10.1063/1.4818950}{{\em J. Math.
  Phys.} {\bfseries 54} (02, 2013) 082201},
  \href{http://arxiv.org/abs/1302.6990}{{\ttfamily arXiv:1302.6990
  [quant-ph]}}.

\end{thebibliography}

\providecommand{\href}[2]{#2}\begingroup\raggedright\endgroup

\end{document}